\documentclass[aps,pra,twocolumn,amsmath,showpacs]{revtex4}
\usepackage{bm}
\usepackage{graphicx}
\usepackage{amsthm}
\usepackage{dsfont}
\usepackage{array}
\usepackage{amsmath}
\usepackage{mathbbol}
\usepackage{epic}
\usepackage{calc}
\usepackage{amsfonts}
\usepackage{amssymb}

\newtheorem{theorem}{Theorem}
\newtheorem*{theorem*}{Theorem}

\newtheorem{corollary}{Corollary}

\providecommand{\openone}{\leavevmode\hbox{\small1\kern-3.8pt\normalsize1}}

\begin{document}
\setlength{\arraycolsep}{2pt}

\title{Quantum steering of multimode Gaussian states by Gaussian measurements: monogamy relations and the Peres conjecture}
\author{Se-Wan Ji$^{1}$, M. S. Kim$^{2,3}$, and Hyunchul Nha$^{1,3}$}
\email{hyunchul.nha@qatar.tamu.edu}
\affiliation{$^1$Department of Physics, Texas A \& M University at Qatar, Doha, Qatar\\
$^2$Quantum Optics and Laser Science Group, Imperial College London, Blackett Laboratory, SW7 2AZ London, UK\\
$^3$School of Computational Sciences, Korea Institute for Advanced Study, Seoul 130-012, Korea\\}

\date{\today}

\begin{abstract}
It is a topic of fundamental and practical importance how a quantum correlated state can be reliably distributed through a noisy channel for quantum information processing. The concept of quantum steering recently defined in a rigorous manner is relevant to study it under certain circumstances and we here address quantum steerability of Gaussian states to this aim. In particular, we attempt to reformulate the criterion for Gaussian steering in terms of local and global purities and show that it is sufficient and necessary for the case of steering a 1-mode system by a $N$-mode system. It subsequently enables us to reinforce a strong monogamy relation under which only one party can steer a local system of 1-mode. Moreover, we show that only a negative partial-transpose state can manifest quantum steerability by Gaussian measurements in relation to the Peres conjecture.
We also discuss our formulation for the case of distributing a two-mode squeezed state via one-way quantum channels making dissipation and amplification effects, respectively. Finally, we extend our approach to include non-Gaussian measurements, more precisely, all orders of higher-order squeezing measurements, and find that this broad set of non-Gaussian measurements is not useful to demonstrate steering for Gaussian states beyond Gaussian measurements. 
\end{abstract}
\maketitle

\providecommand{\openone}{\leavevmode\hbox{\small1\kern-3.8pt\normalsize1}}

\bibliographystyle{apsrev}

\section{Introduction}
Quantum mechanics predicts strong correlations between systems, which may be used as a crucial resource for quantum information processing. The strongest form of correlation is the one that cannot be explained by any local realistic models \cite{Bell}, in which  
the joint probability of obtaining outcomes $a$ and $b$ under local measurements $A$ and $B$, respectively, can be represented by 
\begin{eqnarray}
P_{\rm LHV}(a,b:A,B)=\sum_\lambda p_\lambda P_\lambda(a:A)P_\lambda(b:B),
\end{eqnarray}
with a hidden-variable $\lambda$ and its distribution $p_\lambda$. The nonlocal correlation is important not only from a fundamental point of view but also for practical applications, e.g. unconditionally secure cryptography \cite{Acin} and random-number generation \cite{Pironio}.

During the past decades, there have been numerous theoretical and experimental efforts  demonstrating quantum nonlocality \cite{Brunner}. However, only recently, another form of nonlocal correlation was rigorously defined \cite{Wiseman}, which remarkably addresses the situation envisioned by Einstein, Podolsky, and Rosen (EPR) more closely \cite{EPR}. Namely, in the EPR steering, as put forward by Sch{\" o}dinger \cite{Sch}, Alice performs a local measurement on her system, which can steer Bob's system to a specific ensemble of states, if there exists quantum correlation shared between Alice and Bob. In terms of joint probability distribution, Wiseman {\it et al.} defined the steerability as ruling out the correlations explained by 
\begin{eqnarray}
P_{\rm LHS}(a,b:A,B)=\sum_\lambda p_\lambda P_\lambda(a:A)P_\lambda^Q(b:B),
\end{eqnarray}  
where the superscript $Q$ in $P_\lambda^Q(b:B)$ refers to the probabilities obeying quantum principles. Therefore, by its construction, $P_{\rm LHS}$ is a subset of $P_{\rm LHV}$, which makes it easier to demonstrate quantum steering than quantum nonlocality, even leading to a loophole-free experiment \cite{Smith}. It was recently demonstrated that the observability of EPR steering is intrinsically related to the incompatibility of quantum measurements \cite{Quintino}. In addition, quantum steering is asymmetric by its definition, that is, there can exist cases in which Alice can steer Bob's system but not vice versa \cite{Bowles,Handchen}. 

Quantum steerability is not only of fundamental interest but also of practical merit as it can provide a basis for secure communication. Even when Alice or her measurement devices cannot be trusted at all, Bob can still be convinced of true correlation by ruling out the models in Eq. (2). Indeed, some proposals for one-way device-independent cryptography were recently made based on the quantum steerability \cite{Branciard}. Quantum steering is also relevant to an important, practical, scenario of distributing an entangled state through a noisy channel. For instance, a trusted party Bob prepares an entangled state of $M\times N$ modes and sends the $M$-mode system to Alice through a noisy channel. He then wants to know whether a true correlation has been established between him and Alice despite a possible presence of malicious party that can be disguised in the form of noise. 

In this paper, we study quantum steering of mixed Gaussian states and specifically address the steering criterion in view of local and global purities, which can be formulated by the determinants of the covariance matrix. We show that the determinant condition is sufficient and necessary for the steering of 1-mode system by $N$-mode system through Gaussian measurements. This enables us to derive a strict monogamy relation under which only one party can steer a single-mode Gaussian state \cite{Ji}. However, we present a counter-example of steering a 2-mode system by a single-mode system for which the determinant condition fails to find the steering. Furthermore, we  prove the Peres conjecture \cite{Peres, Pusey} within the Gaussian regime that only a negative partial-transpose state can show quantum steerability by Gaussian measurements. We also discuss the case of distributing a two-mode Gaussian state through a  one-way Gaussian channel giving dissipation or amplification and demonstrating quantum steerability, which is a typical situation of practical relevance for continuous variables. Finally, we extend our approach to include non-Gaussian measurements addressing higher-order squeezing \cite{Hillery} and show that this broad set of non-Gaussian measurements cannot be more useful than Gaussian measurements, which would provide a partial clue to the steerablity of Gaussian states by non-Gaussian measurements.

\section{Preliminaries}
Let us consider a bipartite  Gaussian state $ \rho_{AB} $ of $\left(M + N \right)$-modes which is characterized by its covariance matrix $ \gamma_{AB}$. This covariance matrix can be written in a block form 
\begin{equation}
\label{eq9}
\gamma_{AB}=\left( \begin{array}{cc}
\gamma_{A} & C \\
C^{T} & \gamma_{B} \\
\end{array} \right),
\end{equation}
where $ \gamma_{A} $ and $\gamma_{B}$ are $2M \times 2M $ and $2N \times 2N $ real symmetric positive matrices, respectively. $C$ is a $ 2M \times 2N $ real matrix with its transposed matrix $C^{T}$. A covariance matrix $ \gamma_{AB} $ of a quantum state must satisfy the uncertainty relation as  \cite{Simon}
\begin{equation}
\label{eq10}
\gamma_{AB} + \left( i \Omega_{A} \oplus i \Omega_{B} \right) \geq 0,
\end{equation}  
where $ \Omega_{A\left(B\right)} =\oplus_{i=1}^{M \left(N \right)} \left( \begin{array}{cc}
0 & 1 \\
-1 & 0 \\
\end{array} \right) $. The uncertainty relation in Eq.(\ref{eq10}) implies that $ \gamma_{AB}>0 $ is strictly positive and its symplectic eigenvalues are larger than or equal to unity. 

\section{Gaussian steering: main results}

There are some questions related to quantum steerability, which can be potentially important for quantum communications. (i) One is the monogamy relation among different parties of multi-mode systems \cite{Toner,Brukner}. Specifically, is it possible for Eve to steer Bob's state simultaneously with Alice who steers Bob's system? (ii) There is an issue of one-way steering, namely, there exists a quantum state with which Alice can steer Bob's state while Bob cannot steer Alice's state. This is rather interesting because correlation is typically understood in a bi-directional context. How can we understand this phenomenon of one-way steering in a physical picture? (iii) Peres conjectured that the correlation of a PPT-state, which is non-negative under partial transposition, can always be explained by the LHV models in Eq. (1) \cite{Peres}. Can we shed some lights on this issue by investigating quantum steerability in the context of LHS model in Eq. (2)?

In this work we study the steerability of multi-mode Gaussian systems by Gaussian measurements \cite{Wiseman}. In particular, we aim at addressing steerability in terms of purity, local or global, as all subtle issues of entanglement arise due to mixedness of system induced by, e.g. system-reservoir interaction or inherent experimental noise. We derive the following results.

\subsection{Gaussian steering in view of local and global purities}
\begin{theorem}
If a bipartite Gaussian state $ \rho_{AB} $ of $ \left(M + N\right)$-modes with its covariance matrix $ \gamma_{AB} $ is non-steerable from $A$ to $B$ by Gaussian measurements, the mixedness of the global system $AB$ must be larger than that of the local system $A$, i.e.  
\begin{equation}
\label{eq11}
\det{\gamma_{AB}}\, \geq \det{\gamma_{A}}.
\end{equation}  
Therefore, if $\det{\gamma_{AB}}<\det{\gamma_{A}}$, steering is possible from $A$ to $B$ with Gaussian measurements.
\end{theorem}
\begin{proof}
It is known that a covariance matrix $ \gamma_{AB}$ of a given Gaussian state $ \rho_{AB} $ is non-steerable from $A$ to $B$ with Gaussian measurements iff the following inequality is satisfied \cite{Wiseman},
\begin{equation}
\label{eq12}
{\gamma _{AB}} + {{\bf 0}_A} \oplus i{\Omega _B} \ge 0,
\end{equation}
where $ {\bf 0}_{A} $ is a $ 2M \times 2M $ null matrix. Using an argument based on the Schur complement of $\gamma_{A}$ together with the condition $ \gamma_{A} >0 $ \cite{Horn}, the inequality (\ref{eq12}) is equivalent to
\begin{equation}
\label{eq13}
\left(\gamma_{B} - C^{T} \gamma_{A}^{-1}C \right) + i \Omega_{B} \geq 0.  
\end{equation}
In addition, $ \gamma_{AB} >0$ leads to the positiveness of Schur complement of $\gamma_{A}$, i.e.,
\begin{equation}
\label{eq14}
\gamma_{B}-C^{T} \gamma_{A}^{-1} C >0.
\end{equation}
Eqs. (\ref{eq13}) and (\ref{eq14}) imply that $ \left(\gamma_{B}-C^{T} \gamma_{A}^{-1} C \right) $ can be regarded as a covariance matrix of a quantum state. As a result, all symplectic eigenvalues of $ \left(\gamma_{B}-C^{T} \gamma_{A}^{-1} C \right) $ are larger than or equal to $1$ \cite{Simon}, which implies 
\begin{equation}
\label{eq15}
\det{\left( \gamma_{AB}\right)}=\det{\left( \gamma_{A} \right)} \det{\left( \gamma_{B}-C^{T} \gamma_{A}^{-1}C \right)} \geq \det{\gamma_{A}}.
\end{equation}
\end{proof}

\begin{corollary}
\label{cor1}
For a bipartite Gaussian state of $\left(N + 1 \right) $-modes, the inequality (\ref{eq11}) is sufficient and necessary for non-steerability of Bob's 1-mode system by Alice with Gaussian measurements. 
\end{corollary} 
\begin{proof}
As we see from Eqs. (\ref{eq13}) and (\ref{eq14}), if a given $\left(N + 1 \right)$-modes bipartite Gaussian state is non-steerable from $A$ to $B$ with Gaussian measurements, then $ \left( \gamma_{B} - C^{T} \gamma_{A}^{-1} C \right) $ is a covariance matrix of a physical state and this means all symplectic eigenvalues are larger than or equal to unity. For $\left(N + 1 \right)$-modes, the matrix $ \left( \gamma_{B} - C^{T} \gamma_{A}^{-1} C \right) $ is simply a $ 2 \times 2 $ matrix, so the total number of symplectic eigenvalues is just one. It means that $ \det{\left( \gamma_{B}-C^{T} \gamma_{A}^{-1}C \right)} \geq 1 $ is equivalent to its sympletic eigenvalue larger than or equal to 1, thus satisfying the inequality (\ref{eq13}). Therefore, Eq.(\ref{eq15}) is necessary and sufficient for non-steerability of a $N \times 1 $ Gaussian state $\rho_{AB}$ from $A$ to $B$ with Gaussian measurements.
\end{proof}
\subsection{Monogamy relations}
Based on the above result, we also obtain a strict monogamy relation, which can be a crucial basis for secure quantum communication, e.g., in one-way device-independent cryptography \cite{Branciard}, as follows.
\begin{theorem}
For a tripartite Gaussian state $\rho_{ABE} $ of $\left( M + 1 + N \right)$-modes  with its covariance matrix $ \gamma_{ABE} $, if Alice can steer Bob's state, Eve cannot steer Bob's system simultaneously, which may provide a basis for secure communication using Gaussian steerability. 
\end{theorem}
\begin{proof}
Here we can be restricted to a pure-state of $\left( M + 1 + N \right)$-modes. This is because an arbitrary mixed state of $ \left(M' + 1 + N' \right) $-modes  can be obtained by a partial trace over a pure state. If the part $A$ of $M$-modes cannot steer $B$, nor does $A'$ of $M'$-modes ($M'\le M$) that is obtained by a partial trace. Let us assume that the steering is possible from $A$ to $B$ and from $E$ to $B$ simultaneously. Then by {\bf Corollary 1}, we obtain following inequalities,
\begin{equation}
\label{eq16}
\det{\left(\gamma_{AB} \right)} < \det{\left( \gamma_{A} \right)},\;\; \det{\left( \gamma_{BE} \right)} < \det{\left( \gamma_{E} \right)},
\end{equation}
where $ \gamma_{AB}$, $\gamma_{BE}$, $\gamma_{A}$, and $ \gamma_{E} $ are the covariance matrices of relevant partial states. Because $ \rho_{ABE}$ is a pure state, we have $ \det{\left( \gamma_{AB} \right)} = \det{\left( \gamma_{E} \right)}$ and $\det{\left( \gamma_{A} \right)} = \det{\left( \gamma_{BE} \right)}$. So the two inequalities in Eq. (\ref{eq16}) cannot be satisfied simultaneously. Therefore, a simultaneous steering from both $A$ and $E$ to $B$ is impossible. 
\end{proof}
{\it Remark}. We here note that Reid also presented similar monogamy relations for Gaussian steering recently \cite{Reid, Ficek}. In \cite{Reid}, the situations considered are such that the observables at Bob's station are the same for both pairs of steering test, \{Alice, Bob\} and \{Charlie, Bob\}, whereas Alice and Charlie can take arbitrary observables. On the other hand, there can be cases where Bob's observables correlated with other subsystems vary from one party to another. This is because there does not generally exist a standard form for Gaussian states involving more than two modes. 
An example of 3-mode Gaussian state can be constructed to manifest correlation of $\{X_A^0,X_B^0\}$ and $\{X_A^{\pi/2},X_B^{\pi/2}\}$ between Alice and Bob, and of $\{X_C^{\pi/4},X_B^{\pi/4}\}$ and $\{X_C^{-\pi/4},X_B^{-\pi/4}\}$ between Charlie and Bob, respectively, where $X^\theta\equiv ae^{-i\theta}+a^\dag e^{i\theta}$ is the quadrature amplitude. Let $|\Psi\rangle$ be a two-mode squeezed state with squeezing parameter $r$ and $|0\rangle$ a single-mode vacuum state. Then, consider a 3-mode state given by $\rho=\frac{1}{2}|\Psi\rangle\langle\Psi|_{AB}\otimes |0\rangle\langle0|_C+\frac{1}{2}|0\rangle\langle0|_A\otimes  U_2U_3|\Psi\rangle\langle\Psi|_{BC} U_2^\dag U_3^\dag$, where $U_2$ and $U_3$ are unitary operators yielding a local phase shift of $\frac{\pi}{4}$ on modes 2 and 3, respectively. This state manifests entanglement for two different pairs $\{A,B\}$ and $\{B,C\}$ under the condition $1<\cosh r<3$. Remarkably, the correlated observables are $\{X_A^0,X_B^0\}$ and $\{X_A^{\pi/2},X_B^{\pi/2}\}$ for the pair $\{A,B\}$, and $\{X_C^{\pi/4},X_B^{\pi/4}\}$ and $\{X_C^{-\pi/4},X_B^{-\pi/4}\}$ for the pair $\{C,B\}$, respectively, where $X^\theta\equiv ae^{-i\theta}+a^\dag e^{i\theta}$ is the quadrature amplitude of each mode. The Gaussian version of this non-Gaussian state $\rho$ is readily constructed by only considering the covariance matrix of $\rho$.
On the other hand, our result above proves the strict monogamy relation without any restrictions on the Bob's observables.

{\it Counter-example}: We also remark that the determinant condition in Eq. (5) is only  sufficient for non-steerability of the Gaussian states if the steered subsystem has more than a single-mode. There exist some Gaussian states that satisfy the inequality in Eq. (5), but for which it is possible to steer from $A$ to $B$ with Gaussian measurements.  Let us consider a three mode Gaussian state which has a covariance matrix,
\begin{equation}
\gamma_{A\left(B_1B_2\right)}=\left( \begin{array}{cccccc}
2 & 0 & 1.88 & 0 & 0.37 & 0 \\
0 & 2 & 0 & -0.39 & 0 & -0.71 \\
1.88 & 0 & 2.78 & 0 & 0 & 0 \\
0 & -0.39 & 0 & 2.78 & 0 & 0\\
0.37 & 0 & 0 & 0 & 1.14 & 0 \\
0 & -0.71 & 0 & 0 & 0 & 1.14 \\
\end{array} \right),
\end{equation}
where the system $B$ is composed of two modes $B_1B_2$. One finds that this state satisfies the inequality in Eq. (5), i.e., $ \det \left(\gamma_{A\left(B_1B_2\right)} \right) \approx 9.187 > \det \left(\gamma_A \right) =4 $, but it can be possible to steer from $A$ to $B$, since $\det \left( \gamma_{A\left(B_1B_2\right)} + \left( {\bf 0}_A \oplus i \Omega_{B_1B_2} \right) \right) \approx -1.958 < 0$. 

We may obtain another form of monogamy relation bearing on permutation symmetry as follows.
\begin{theorem}
Let us consider a tripartite Gaussian state $ \rho _{ABC}$ of $ \left(L + M + N \right) $ modes. If the state is invariant under exchanging the parties $A$ and $B$ ($L=M$), it is impossible to steer from $A$ to $C$ and from $B$ to $C$ simultaneously by Gaussian measurements. 
\end{theorem}

\begin{proof}
Let us consider a covariance matrix $ \gamma _{ABC} $ which has the form
\begin{equation}
\label{eq17}
\gamma _{ABC} = \left( \begin{array}{ccc}
\gamma_A & C_1 & C_2 \\
C_1 ^T & \gamma_A & C_2 \\
C_2 ^T & C_2 ^T & \gamma_C \\
\end{array} \right),
\end{equation}
\\
where $ \gamma_A $, $ \gamma_C $ are $ 2M \times 2M $, $ 2N \times 2N $ real  symmetric positive matrices, and $ C_1$, $C_2$ are $ 2M \times 2M $, $ 2M \times 2N$ matrices. A tripartite $ \left( M + M + N \right) $-modes Gaussian state $\rho _{ABC} $, which is invariant under exchanging the modes $A$ and $B$, has the covariance matrix in this form. The covariance matrix must satisfy the uncertainty relation,
\begin{equation}
\label{eq18}
\gamma _{ABC} + i \left( \Omega ^{(2M)} \oplus \Omega ^{(2M)} \oplus \Omega ^{(2N)} \right) \geq 0,
\end{equation}
where $ \Omega ^{(2k)} = \mathop  \oplus \limits_{i = 1}^{k} \left( \begin{array}{cc}
0 & 1 \\
-1 & 0 \\
\end{array} \right) $. With the local uncertainty relation $ \gamma_C +i \Omega ^{(2N)} \geq 0 $, the Schur complement of $\gamma_C +i \Omega ^{(2N)}$ should be positive semidefinite. This means 
\begin{widetext}
\begin{equation}
\label{eq19}
ABC / C =\left( \begin{array}{cc}
\gamma_A + i\Omega ^{(2M)} -C_2 \left( \gamma_C + i \Omega ^{(2N)} \right)^{-1}C_2^{T} & C_1 -C_2 \left( \gamma_C + i \Omega ^{(2N)} \right)^{-1}C_2^{T} \\ C_1^{T} -C_2 \left( \gamma_C + i \Omega ^{(2N)} \right)^{-1}C_2^T & \gamma_A + i\Omega ^{(2M)} -C_2 \left( \gamma_C + i \Omega ^{(2N)} \right)^{-1}C_2^{T}
\end{array} \right) \geq 0 ,
\end{equation}  
\end{widetext}
where $ ABC / C $ is the Schur complement of the matrix $ \gamma_C +i \Omega ^{(2N)} $ with respect to the total matrix $\gamma _{ABC} + i \left( \Omega ^{(2M)} \oplus \Omega ^{(2M)} \oplus \Omega ^{(2N)} \right)$.\\

If Alice can steer Charlie's state, we have
\begin{equation}
\label{eq20}
\gamma _{AC} + {\bf 0}_A \oplus i \Omega_C =\left( \begin{array}{cc}
\gamma_A & C_2 \\
C_2 & \gamma_C +i \Omega_C 
\end{array} \right) \not\ge 0.
\end{equation}  
The Schur complement of $\gamma_C + i \Omega_C\ge 0 $ with respect to the above matrix is then not positive semidefinite, i.e.
\begin{equation}
\label{eq21}
\gamma_A - C_2 \left( \gamma_C +i \Omega _C \right) C_2^{T} \not\ge 0.
\end{equation} 
Thus there exists at least one complex vector $ \vec{a}_0 \in \mathbb{C}^{2M}$ satisfying
\begin{equation}
\label{eq22}
\vec{a}_0 ^{\dagger} \left( \gamma_A - C_2 \left( \gamma_C +i \Omega _C \right) C_2^{T} \right) \vec{a}_0 <0.
\end{equation}
Now turning attention back to Eq.(\ref{eq19}), all complex vectors $ \vec{J} \in \mathbb{C}^{4M} $ should satisfy
\begin{equation} 
\label{eq23}
\vec{J}^{\dagger} \left( ABC/C \right) \vec{J} \geq 0.
\end{equation}
This must also be true for the choice of $ \vec{J} = \left( \begin{array}{c}
\vec{a}_0 \\
\vec{a}_0 \\
\end{array} \right)$, which gives 
\begin{eqnarray}
\label{eq24}
0 &&\leq\vec{J}^{\dagger} \left( ABC/C \right) \vec{J}\nonumber\\
&&\leq 4 \vec{a}_{0}^{\dagger} \left( \gamma_A -C_2 \left( \gamma_C + i \Omega ^{(2N)} \right)^{-1}C_2^{T} \right) \vec{a}_{0},\\
\end{eqnarray}
where the last line is obtained using an inequality $ \left( \begin{array}{cc}
\gamma_A & -C_1 \\
-C_1 ^{T} & \gamma_A
\end{array} \right) - i \Omega^{\left(2M \right)} \geq 0 $ equivalent to $ \left( \begin{array}{cc}
\gamma_A & C_1 \\
C_1 ^{T} & \gamma_A 
\end{array} \right) +i\Omega^{\left(2M \right)} \geq 0 $. However, if both $A$ and $B$ can steer $C$ simultaneously, there occurs a contradiction, i.e., the last term of Eq.(\ref{eq24}) should be negative from Eq.(\ref{eq22}). 
\end{proof}

\subsection{Peres conjecture}
Next we turn our attention the Peres conjecture \cite{Peres}. A stronger Peres conjecture on quantum steering was made \cite{Pusey}, which is recently disproved \cite{Moroder}, together with the original conjecture on nonlocality \cite{Vertesi}. However, we show that the Peres conjecture is valid in the Gaussian steering scenario.
\begin{theorem}
If a given bipartite Gaussian state $ \rho _{AB} $ of $\left(M + N \right)$-modes is steerable from $A$ to $B$ or from $B$ to $A$ with Gaussian measurements, then this state must have at least one negative eigenvalue under partial transposition, i.e., NPT.
\end{theorem}
\begin{proof}
If the covariance matrix of a Gaussian state violates the condition Eq. (\ref{eq12}), i. e., $ {\gamma _{AB}} + {{\bf 0}_A} \oplus i{\Omega _B} \not\ge 0$, the steering from $A$ to $B$ is possible by Gaussian measurements. In this case, the Schur complement of $ \gamma_{B} + i{\Omega _B}\ge0$ is not positive semidefinite. In other words, there exists a complex vector $ {\vec a_0}$ satisfying 
\begin{equation}
\label{eq25}
\vec a_0^\dag \left( {\gamma_{A} - C{{\left[ { \gamma_{B} + i{\Omega _B}} \right]}^{ - 1}}{C^T}} \right){\vec a_0} < 0.
\end{equation}
Since $ \gamma_{AB} $ is also a covariance matrix, the following inequality must be satisfied for an arbitrary complex vector $ \vec{a}$,
\begin{equation}
\label{eq26}
\vec{a}^{\dagger} \left( {\gamma_{A} + i{\Omega _A} - C{{\left[ {\gamma_{B} + i{\Omega _B}} \right]}^{ - 1}}{C^T}} \right){\vec a} \ge 0
\end{equation}
Combining Eqs.(\ref{eq25}) and (\ref{eq26}), we see that the complex vector $ {\vec a_0}$ satisfies 
\begin{equation}
\label{eq27}
0 < \; - \vec a_0^\dag \left( {\gamma_{A} - C{{\left[ {\gamma_{B} + i{\Omega _B}} \right]}^{ - 1}}{C^T}} \right){\vec a_0}\;\; \leq \;\;\vec a_0^\dag \left( {i{\Omega _A}} \right){\vec a_0}.
\end{equation}
\\
If a bipartite quantum state is separable, then its density matrix has to be positive semidefinite under partial transposition \cite{Peres}. For a Gaussian state, this can be formulated using its covariance matrix \cite{Werner} as
\begin{equation}
\label{eq28}
{ \gamma _{AB}} + i\left( { - {\Omega _A} \oplus {\Omega _B}} \right) \ge 0.
\end{equation}
The violation of the inequality (\ref{eq28}) means that the density matrix of a given bipartite Gaussian state has at least one negative eigenvalue under partial transposition. This is called non-positive partial transposition (NPT) criterion for inseparability. Eq. (\ref{eq28}) can be reexpressed using Schur complement as
\begin{equation}
\label{eq29}
\gamma_{A} - i{\Omega _A} - C{\left( {\gamma_{B} + i{\Omega _B}} \right)^{ - 1}}{C^T} \ge 0
\end{equation}
Now, let us consider a bipartite Gaussian quantum state with which $A$ can steer $B$ by Gaussian measurements. Then there exists a complex vector satisfying Eqs. (\ref{eq25}) and (\ref{eq27}), which leads to
\begin{equation}
\label{eq30}
\vec a_0^\dag \left( {\gamma_{A} - i{\Omega _A} - C{{\left[ {\gamma_{B} + i{\Omega _B}} \right]}^{ - 1}}{C^T}} \right){\vec a_0}\; < \;\,0.
\end{equation}
Therefore, a bipartite Gaussian state which is steerable by Gaussian measurements violates Eq. (\ref{eq29}), which means NPT. 
\end{proof}

\section{one-channel decoherence}
Let us now consider a realistic problem of distributing an entangled state through a noisy channel. In particular, we consider the case that one party prepares an entangled state (two-mode squeezed state) and sends a subsystem to the other through a decoherence channel (dissipating or amplifying). 

Under these circumstances, we may also discuss one-way steering using our Theorem 1 and Corollary 1. If the steering is possible only from $A$ to $B$, not from $B$ to $A$, the local and global purities show the relation $\det{\gamma_{B}}\leq\det{\gamma_{AB}}<\det{\gamma_{A}}$. We illustrate the cases of one-way steering by a two-mode pure state undergoing one-channel decoherence via dissipation and amplification, respectively.\\

A covariance matrix $ \gamma_{AB}$ of a two-mode squeezed vacuum state $\rho_{AB} $ is given by
\begin{equation}
\gamma_{AB}=\left( \begin{array}{cccc}
\cosh{2r} & 0 & \sinh{2r} & 0 \\
0 & \cosh{2r} & 0 & -\sinh{2r} \\
\sinh{2r} & 0 & \cosh{2r} & 0 \\
0 & -\sinh{2r} & 0 & \cosh{2r}  \\
\end{array} \right),
\end{equation}
where $r$ is the squeezing parameter. For any $ r>0$, the given state is entangled and steerable in both ways. Let us assume that the mode $B$ goes through a lossy channel which is characterized by a parameter $ \eta $. That is, the output state will have a covariance $\gamma_{AB}\rightarrow\gamma_{AB'}$, 
\begin{widetext}
\begin{equation}
\gamma_{AB'}=\left( \begin{array}{cccc}
\cosh{2r} & 0 & \sqrt{\eta}\sinh{2r} & 0 \\
0 & \cosh{2r} & 0 & -\sqrt{\eta}\sinh{2r} \\
\sqrt{\eta}\sinh{2r} & 0 & \eta \cosh{2r}+1-\eta & 0 \\
0 & -\sqrt{\eta}\sinh{2r} & 0 & \eta\cosh{2r}+1-\eta \\
\end{array} \right),
\end{equation}
\end{widetext}
where $\eta$ ($ 0 \leq \eta \leq 1 $) characterizes the degree of dissipation, with $\eta\rightarrow0$ the case of decohering mode $B$ to a vacuum state. 

Now, let us look into the steerability from $A$ to $B$ or from $B$ to $A$ with Gaussian measurements. It is straightforward to find that the steering is always possible from $A$ to $B$ except $\eta=0$, while the steering is possible from $B$ to $A$ only in the case $ \frac{1}{2} < \eta \leq 1 $. 

It might seem to suggest that the subsystem undergoing a noisy channel becomes harder to steer the other party than the other way around \cite{Cwlee}. However, it is not quite true.  Now let us consider the case of mode $B$ undergoing an amplifying channel in comparison. Then the covariance matrix $\gamma_{AB}$ is changed to $\gamma_{AB''}$ as 
\begin{widetext}
\begin{equation}
\gamma_{AB''}=\left( \begin{array}{cccc}
\cosh{2r} & 0 & \sqrt{G}\sinh{2r} & 0 \\
0 & \cosh{2r} & 0 & -\sqrt{G}\sinh{2r} \\
\sqrt{G}\sinh{2r} & 0 & G\cosh{2r}+G-1 & 0 \\
0 & -\sqrt{G}\sinh{2r} & 0 & G\cosh{2r}+G-1 \\
\end{array} \right),
\end{equation}
\end{widetext}
where $G \geq 1 $ is the gain parameter of the amplifying channel. In this case, the steering is possible from $A$ to $B$ only when $ 1 \leq G < \frac{2\cosh{2r}}{\cosh{2r}+1}$ is satisfied. In other words, if the gain is larger or equal to $\frac{2\cosh{2r}}{\cosh{2r}+1} \left( < 2\right)$, it is impossible to steer from $A$ to $B$ by Gaussian measurements. On the other hand, it is always possible to steer from $B$ to $A$ (since we assume $r>0$).  

The above two examples of dissipating and amplification channels again suggest that the local mixedness compared with the global mixedness can be an important factor affecting the evolution of quantum correlation, precisely, the quantum steering. As a final remark, once the state under a noisy channel becomes unsteerable, say from Alice to Bob, it remains so even though the trusted party, Bob, performs a local pre-processing \cite{Quintinno}. For instance, under the amplification of Bob's mode above the gain $G=\frac{2\cosh{2r}}{\cosh{2r}+1}$, Alice can no longer steer Bob's local state by Gaussian measurements even if Bob attempts to attenuate his mode after the amplification, which only affects the steerability by Bob on Alice's mode in a negative manner. However, it is worthwhile to more thoroughly investigate the steerability under local filtering, which may project the output state to a non-Gaussian one (e.g. projection to a qubit space). It is still an open question whether the steeering by Gaussian measurements is both sufficient and necessary for the steerability of Gaussian states. Therefore, although an initial state is non-steerable by Gaussian measurements, it might be steerable by non-Gaussian measurements. In such a case, the non-Gaussian state by a local processing might be steerable by a set of non-Gaussian measurements, which would then demonstrate the incompleteness of Gaussian measurements for steering of Gaussian states. This issue is beyond the scope of current work and will be addressed elsewhere.

\section{Summary}

In this work, we have studied quantum steering of Gaussian states by Gaussian measurements, specifically by looking into local and global purities quantified by the determinants of covariance matrix. We have proved that the determinant condition Eq. (5) is necessary and sufficient for Gaussian steerability when a $N$-mode subsystem steers a single-mode system. On the other hand, we have also constructed a counter-example of steering a two-mode system by a single-mode system in which the determinant condition fails to describe the quantum steering completely. Moreover, we have shown that only a negative partial-transpose state can manifest quantum steerability by Gaussian measurements in relation to the Peres conjecture.

Nevertheless, the mixedness relation may be a crucial tool to address the entanglement dynamics in practically relevant situations. For instance, using the determinant condition, we have proved a strict monogamy relation under which only one party can steer a single-mode Gaussian state. As another practically important situation, we have considered the case of distributing a two-mode entangled state through one-way dissipating and amplifying channels, respectively, and showed how one-way steering may arise.  

Our work is here restricted to performing Gaussian measurements for steerability, which must be extended to non-Gaussian measurements in order to have a more complete understanding of quantum steering even for Gaussian states, e.g. identification of genuine one-way steering. As an extension, we have considered higher-order quadrature measurements in Appendix and shown that if Gaussian steering criterion fails to detect steerability of Gaussian state, nor does such high-order non-Gaussian measurements. We will incorporate a broader class of non-Gaussian measurements to address quantum steerability of Gaussian states in future.

\section{Acknolwedgements}
This work was supported by the NPRP grant 4-520-1-083 from Qatar National Research Fund.

{\it Note Added:} A recent article \cite{Adesso} obtained similar results reported in the present paper and \cite{Ji}, such as the role played by local vs. global purities and the Peres conjecture.

\section{Appendix}

\subsection{Reid's criterion}
Let us consider three observables $ \{ \hat{B}_1, \hat{B}_2, \hat{B}_3 \}$ obeying a commutation relation $ [ \hat B_1,\; \hat B_2 ] = i \hat B_3 $. It gives a Heisenberg uncertainty relation, 
\begin{equation}
\label{uncer}
\Delta^2 \left( \hat B_1 \right)_{\rho}  \Delta^2 \left( \hat B_2 \right)_{\rho} \geq \frac{1}{4} | \left< \hat B_3 \right>_{\rho} |^2,
\end{equation}
where $ \Delta^2\left( \hat B_i \right)_{\rho} $ and $ \left< \hat B_i \right>_{\rho} $ are the variance and the mean value, respectively, of the observable $ \hat B_i$ ($i=1,2,3$) for a quantum state $ \rho $. When a bipartite system can be described by the LHS model, we obtain a non-steerability condition as
\begin{equation}
\label{inf-uncer}
\Delta^2_{inf}\left( \hat B_1 \right) \Delta^2_{inf}\left( \hat B_2 \right)  \geq \frac{1}{4} \left|\left< \hat B_3 \right>\right|^2_{inf},
\end{equation}
where $ \Delta^2_{inf} \left( \hat B_i \right) $ is an {\it inferred} variance ($i=1,2$) \cite{reid1} and 
\begin{eqnarray}
\label{inf-avg}
 \left|\left< \hat B_3 \right>\right|_{inf} = \sum_{a_3} P\left( A_3 =a_3 \right) \left| \left< \hat B_3 \right>_{A_3=a_3} \right|,
\end{eqnarray}
which are all conditioned on the measurement outcomes for the subsystem $A$ \cite{cavalcanti}. 

The inferred variance of $ \hat B_i$ with its estimate $ B_{i,est} \left(A_i \right) $ based on the outcome $A_i$ is defined by \cite{reid1, cavalcanti}
\begin{equation}
\label{infer}
\Delta^2_{inf} \left( \hat B_i \right) = \left< \left[ \hat B_i - B_{i,est} \left( A _i \right) \right]^2 \right>. 
\end{equation}
With a choice of linear estimate $ B_{i,est} \left(A_i \right)=-g_i A_i+ \left< \hat B_i + g_i A_i \right> $  ($ g_i$: a real number) \cite{cavalcanti, reid2}, one can derive the optimal  inferred variance $\Delta^2_{inf}\left( \hat B_i \right)_{opt}$ \cite{reid2},
\begin{equation}
\label{op-infer}
\Delta^2_{inf}\left( \hat B_i \right)_{opt}= \Delta^2 \left( \hat B_i \right) -\frac{\left< A_i \hat B_i \right>-\left<A_i \right> \left< \hat B_i \right>}{\Delta^2 \left(A_i \right)}.
\end{equation}
If a given quantum state violates the inequality (\ref{inf-uncer}), the  state is steerable from Alice to Bob. 

\subsection{Two-mode Gaussian states}
A standard form of a covariance matrix for a two-mode state is given by,
\begin{equation}
\label{standard}
\gamma _{AB} ^{(s)}  = \left( \begin{array}{cc}
 \gamma_A & \tilde C \\ 
 \tilde C^T & \gamma_B \\ 
 \end{array} \right) = \left( \begin{array}{cccc}
 a & 0 & c_1 & 0 \\ 
 0 & a & 0 & c_2  \\ 
 c_1 & 0 & a & 0 \\ 
 0 & c_2 & 0 & a \\ 
 \end{array} \right).
\end{equation} 
Without loss of generality, we may set $ c_1 \geq \left| c _2 \right| \geq 0 $. For a Gaussian state $ \rho _{AB}$ with a covariance matrix in the standard form (\ref{standard}) and the means $ \left< \hat a \right>=\left< \hat b \right>=0 $, the normally-ordered characteristic function $C_N \left( \alpha _1, \alpha _2 \right)$, which facilitates the calculation of various moments of our interest, is given by \cite{mista}
\begin{widetext}
\begin{equation}
\label{characteristic}
C_N \left( \alpha _1, \alpha _2 \right)  = {\rm Tr} \left[ \rho _{AB} e^{\alpha _1 \hat a^{\dagger}+ \alpha _2 \hat b^{\dagger}}e^{-\alpha ^{*}_{1}\hat a-\alpha ^{*}_2 \hat b } \right] =\exp \left[ - \left(A \left| \alpha _1 \right|^2 + B \left| \alpha _2 \right|^2 \right) 
+ \left(C \alpha ^{*}_{1} \alpha ^{*}_{2} + D \alpha _1 \alpha ^{*}_2 + C \alpha _1 \alpha _2 + D \alpha ^{*}_1 \alpha _2 \right)  \right], 
\end{equation} 
\end{widetext}
where $ A = \left< \hat a^{\dagger} \hat a \right>=\frac{\left(a-1 \right)}{2} $, $B=\left< \hat b^{\dagger} \hat b \right>=\frac{\left(b-1 \right)}{2}$, $C=\left< \hat a \otimes \hat b \right>=\frac{\left( c_1-c_2 \right)}{4}$, and $D= -\left< \hat a^{\dagger} \otimes \hat b \right>=\frac{-\left(c_1+c_2 \right)}{4}$. If the Gaussian state is non-steerable from $A$ to $B$ by Gaussian measurements, we have the determinant condition
\begin{widetext}
\begin{equation}
\label{standard-steering}
\det \left( \gamma _{AB} \right) \geq \det\left( \gamma_A \right) \Leftrightarrow \left( \left[ 2A +1 \right]\left[ 2B+1 \right] - 4\left[C-D\right]^2\right)\left( \left[ 2A +1 \right]\left[ 2B+1 \right] -4\left[C+D \right]^2 \right) \geq \left(2A +1 \right) ^2.
\end{equation}
\end{widetext}
Using $\alpha^2+\beta^2\ge 2\alpha\beta$, where $\alpha$ and $\beta$ are the two terms of the product in the left-hand side of (\ref{standard-steering}), one can readily obtain two equivalent inequalities as
\begin{equation}
\label{arithmetic1}
a^2\left(b-1\right)^2 \geq \frac{1}{4}\left( c^2_1 + c^2_2 \right)^2 \Leftrightarrow AB \geq \left( C^2 + D^2 \right)-\frac{B}{2}.
\end{equation}

\subsection{non-Gaussian measurements}
We now consider two non-commuting observables in high orders, i.e. 
\begin{equation}
\label{conjugate1}
\begin{array}{cc}
\hat X^{\left(N\right)}_{A\left(B\right)} =\frac{1}{\sqrt{2}}\left( \hat a^N \left( \hat b^{N} \right)+ \hat a^{\dagger N} \left( \hat b^{\dagger N} \right) \right), \\
\hat P^{\left(N\right)}_{A\left(B\right)} =\frac{1}{\sqrt{2}i}\left( \hat a^N \left( \hat b^{N} \right)- \hat a^{\dagger N} \left( \hat b^{\dagger N} \right) \right).
\end{array}
\end{equation}
For $N=1$, the operators in Eq.(\ref{conjugate1}) are two regular quadrature amplitudes while $N>1$ defines higher-order quadratures. For example, Hillery considered the case of $N=2$ for the study of higher-order squeezing effect \cite{Hillery}. We show that if a two-mode Gaussian state with the covariance matrix in Eq.(\ref{standard}) and normal characteristic function in Eq.(\ref{characteristic}) is non-steerable by Gaussian measurements from Alice to Bob, i.e. satisfying the inequalities (\ref{standard-steering}) and (\ref{arithmetic1}), then the non-steerability criteria (\ref{inf-uncer}) are also satisfied with the operators in Eq.(\ref{conjugate1}) for all $N\geq 2 $. This would give a partial clue for the steerablity of Gaussian states with non-Gaussian measurements. 

Let us first consider the case of $N=2$. We obtain the non-steerability condition in Eq. (\ref{inf-uncer}), using Eqs. (\ref{inf-avg}) and (\ref{op-infer}), as 
\begin{equation}
\label{n2}
\left( b^2 +1 -\frac{\frac{1}{4}\left( c^2_1 +c^2_2 \right)^2}{a^2+1} \right)\left( b^2 +1 -\frac{c^2_1c^2_2}{a^2+1} \right) \geq 4b^2 .
\end{equation} 
If a Gaussian state satisfies the inequalities in Eqs. (\ref{standard-steering}) and (\ref{arithmetic1}), it also satisfies the inequality (\ref{n2}). 
\begin{proof}
\begin{equation}
\label{pn2}
\begin{array}{llll}
\left( b^2 +1 -\frac{\frac{1}{4}\left( c^2_1 +c^2_2 \right)^2}{a^2+1} \right)\left( b^2 +1 -\frac{c^2_1c^2_2}{a^2+1} \right) - 4b^2 \\
\geq \left( b^2+1 -\frac{\frac{1}{4}\left( c^2_1 +c^2_2 \right)^2}{a^2+1} \right)^2-4b^2 \\
\begin{array}{ll}= \frac{1}{a^2+1}\left[ a^2\left(b-1\right)^2-\frac{1}{4}\left( c^2_1 +c^2_2 \right)^2+\left(b-1\right)^2 \right] \\
\hspace{0.4cm}\times \left[ \left( b+1 \right)^2 -\frac{\frac{1}{4}\left( c^2_1 +c^2_2 \right)^2}{a^2+1} \right]\geq 0, \end{array} \\
\end{array}
\end{equation}
\end{proof}
using $ \frac{1}{4}\left( c_1^2 + c_2^2
 \right)^2 -c_1^2c_2^2=\frac{1}{4}\left( c_1^2 - c_2^2 \right)^2 \geq 0$ and the first condition in Eq.(\ref{arithmetic1}).

To prove the satisfaction of other inequalities in Eq. (\ref{inf-uncer}) for $N > 2$, we use the technique of mathematical induction. Let us assume that the following inequality for $N$ is satisfied,
\begin{equation}
\label{nn}
\begin{array}{ll}
B^N\left[\left(A+1\right)^N+A^N\right]\geq 2\left[ C^N\pm\left(-1\right)^ND^N \right]^2, 
\end{array}
\end{equation}
which can be shown to lead to the satisfaction of
\begin{equation}
\label{nnn}
\begin{array}{ll}
\Delta^2_{inf} \left( \hat X^{\left(N\right)}_B \right)\Delta^2_{inf} \left( \hat P^{\left(N\right)}_B \right) \geq \frac{1}{4}\left| \left< \left[ \hat X^{\left(N\right)}_B, \hat P^{\left(N\right)}_B \right] \right> \right|^2_{inf} .
\end{array}
\end{equation} 
In Eq.(\ref{nnn}), we use the expressions
\begin{widetext}
\begin{equation}
\label{nn-n}
\begin{array}{ccc}
2\Delta^2_{inf} \left( \hat X^{\left(N\right)}_B \right) = N!\left(B^N+\left(B+1\right)^N\right)-\frac{4N!\left(C^N+(-1)^ND^N\right)^2}{N!\left(A^N+\left(A+1\right)^N\right)}\\
2\Delta^2_{inf} \left( \hat P^{\left(N\right)}_B \right) = N!\left(B^N+\left(B+1\right)^N\right)-\frac{4N!\left(C^N+(-1)^{\left(N+1\right)}D^N\right)^2}{N!\left(A^N+\left(A+1\right)^N\right)} \\
\left| \left< \left[ \hat X^{\left(N\right)}_B, \hat P^{\left(N\right)}_B \right] \right> \right|^2_{inf} =\left[ N! \left( \left(B+1\right)^N-B^N\right) \right]^2
\end{array}
\end{equation}
\end{widetext}
Due to the condition $ c_1 \geq \left|c_2 \right| \geq 0$, we have $\left[ C^{N+1}+\left(-1\right)^{N+1}D^{N+1} \right]^2 \geq \left[ C^{N+1}+\left(-1\right)^{N+2}D^{N+1} \right]^2$. 
We then need to show that the inequality (\ref{nn}) is also safisfied for $N+1$ with a positive sign only, i.e. $B^{N+1}\left[\left(A+1\right)^{N+1}+A^{N+1}\right]\geq 2\left[ C^{N+1}+\left(-1\right)^{N+1}D^{N+1} \right]^2 $. \\
\begin{proof} 
\begin{equation}
\label{nn+1}
\begin{array}{llllll}
B^{N+1}\left[\left(A+1\right)^{N+1}+A^{N+1}\right]- 2\left[ C^{N+1}+\left(-1\right)^{N+1}D^{N+1} \right]^2 \\
{\begin{array}{ll} = BA B^N\left[\left(A+1\right)^N+A^N\right] +B^{N+1}\left(A+1\right)^{N}\\
\;\;-2\left[ C^{N+1}+\left(-1\right)^{N+1}D^{N+1} \right]^2 \\ \end{array} } \\
\begin{array}{ll} \geq \left[\left( C^2 +D^2 \right)-\frac{B}{2}\right]B^N\left[\left(A+1\right)^N+A^N\right] \\
\;\; +B^{N+1}\left(A+1\right)^{N}-2\left[ C^{N+1}+\left(-1\right)^{N+1}D^{N+1} \right]^2 \\ \end{array} \\
\begin{array}{ll} \geq 2\left(C^2 +D^2\right)\left[ C^N+\left(-1\right)^ND^N \right]^2 + \frac{B^{N+1}}{2}\left[ \left(A+1\right)^N-A^N \right] \\
\;\;-2\left[ C^{N+1}+\left(-1\right)^{N+1}D^{N+1} \right]^2 \\ \end{array} \\
\geq 4\left[ \left(-1\right)^NC^ND^N\left(C^2+D^2\right)\right]+\frac{B^{N+1}}{2}\left[ \left(A+1\right)^N-A^N \right]. \\
\geq 0
\end{array}
\end{equation}
\end{proof}
The first inequality is due to the non-steering condition in Eq. (\ref{arithmetic1}), the second from the assumption in Eq.(\ref{nn}), the third due to the relation of  arithmetical-geometric means after expanding terms, and the fourth from the assumption $ c_1 \geq \left|c_2\right| \geq0$ implying $ \left(-CD \right)^N \geq 0$.

\end{document}